\newif\ifdraft \drafttrue
\newif\iffull \fulltrue
\newif\ifnips \nipsfalse
\makeatletter \@input{tex.flags} \makeatother

\documentclass[11pt]{article}

\title{Predicting with Distributions}

\usepackage{kpfonts}
\usepackage{fullpage}

\usepackage{amsmath, amssymb, amsthm}

\usepackage{algorithm}
\usepackage[noend]{algorithmic}

\usepackage{mathtools}
\usepackage{verbatim}
\usepackage{color}

\usepackage{xcolor}

\usepackage{multicol}

\usepackage{pgf}
\usepackage{tikz}
\usetikzlibrary{arrows,automata}
\usepackage[latin1]{inputenc}
\usepackage{verbatim}

\definecolor{DarkGreen}{rgb}{0.1,0.5,0.1}
\definecolor{DarkRed}{rgb}{0.5,0.1,0.1}
\definecolor{DarkBlue}{rgb}{0.1,0.1,0.5}
\usepackage[]{hyperref}
\hypersetup{
    unicode=false,          
    pdftoolbar=true,        
    pdfmenubar=true,        
    pdffitwindow=false,      
    pdftitle={},    
    pdfauthor={}
    pdfsubject={},   
    pdfnewwindow=true,      
    pdfkeywords={keywords}, 
    colorlinks=true,       
    linkcolor=DarkRed,          
    citecolor=DarkGreen,        
    filecolor=DarkRed,      
    urlcolor=DarkBlue,          
}
\usepackage{xspace}
 \usepackage{cleveref}
\usepackage{thmtools, thm-restate}

\newcommand{\sw}[1]{\ifdraft \textcolor{blue}{[Steven: #1]}\fi}
\newcommand{\mk}[1]{\ifdraft \textcolor{red}{[Michael: #1]}\fi}

\newcommand\bZ{\mathbf{Z}}

\newcommand\NN{\mathbb{N}}
\newcommand\RR{\mathbb{R}}
\newcommand\cT{\mathcal{T}}
\newcommand\cA{\mathcal{A}}

\newcommand\cE{\mathcal{E}}

\newcommand\cL{\mathcal{L}}

\newcommand\cH{\mathcal{H}}

\newcommand\cP{\mathcal{P}}

\newcommand\cX{\mathcal{X}}
\newcommand\cY{\mathcal{Y}}
\newcommand\cC{\mathcal{C}}
\newcommand\cD{\mathcal{D}}

\newcommand\cN{\mathcal{N}}

\DeclareMathOperator{\poly}{poly}

\DeclareMathOperator*{\Expectation}{\mathbb{E}}
\newcommand{\Ex}[2]{\Expectation_{#1}\left[#2\right]}

\newcommand{\loss}{\text{loss}}
\newcommand{\gen}{\text{Gen}}
\newcommand{\genn}{\text{Gen}'}
\newcommand{\lab}{\mathbf{Lab}}
\newcommand{\pwd}{\text{PwD}}
\newcommand{\eps}{\varepsilon}
\def\epsilon{\varepsilon}

\newtheorem*{theorem*}{Theorem}

\declaretheorem[
  name=Theorem,
  refname={theorem, theorems},
  Refname={Theorem, Theorems}]{theorem}
\declaretheorem[
  name=Lemma,
  refname={lemma, lemmas},
  Refname={Lemma, Lemmas}]{lemma}

\declaretheorem[
  name=Claim,
  refname={claim, claims},
  Refname={Claim, Claims}]{claim}

\declaretheorem[
  name=Definition,
  refname={definition, definitions},
  Refname={Definition, Definitions}]{definition}
\declaretheorem[
  name=Assumption,
  refname={assumption, assumptions},
  Refname={Assumption, Assumptions}]{assumption}

\usepackage{natbib}

\author{Michael Kearns\thanks{
Dept. of Computer and Information Sciences, University of Pennsylvania. Email:~\href{mailto:mkearns@cis.upenn.edu}{mkearns@cis.upenn.edu} }
 \and Zhiwei Steven Wu\thanks{
Dept. of Computer and Information Sciences, University of Pennsylvania. Email:~\href{mailto:wuzhiwei@cis.upenn.edu}{wuzhiwei@cis.upenn.edu} }}

\begin{document}

\maketitle

\begin{abstract}
  We consider a new learning model in which a joint
  distribution over vector pairs $(x,y)$ is determined by an unknown
  function $c(x)$ that maps input vectors $x$ not to individual
  outputs, but to entire {\em distributions\/} over output vectors
  $y$.  Our main results take the form of rather general reductions
  from our model to algorithms for PAC learning the function class and
  the distribution class separately, and show that virtually every
  such combination yields an efficient algorithm in our model.  Our
  methods include a randomized reduction to classification noise and
  an application of Le Cam's method to obtain robust learning
  algorithms.
\end{abstract}


\newcommand{\lp}{\cL_\cP}
\newcommand{\Mp}{m_{\cP}}
\newcommand{\lc}{\cL_\cC}
\newcommand{\lm}{\cL_M}
\newcommand{\Mc}{m_{\cC}}
\newcommand{\ml}{\mathbf{ML}}
\newcommand{\kl}{ \mathrm{KL}}
\newcommand{\EX}{ \mathit{EX}}
\renewcommand{\widehat}{\hat}
\renewcommand{\cH}{\cC}

\section{Introduction}

We consider a new variant of the {\em Probably Approximately Correct
  (PAC)} learning framework. In our model, a joint
distribution over vector pairs $(x,y)$ is determined by an unknown
target function $c(x)$ that maps input vectors $x$ not to individual
outputs, but to entire {\em distributions\/} over output vectors $y$
in some large space. This model generalizes settings such as learning
with classification noise or errors, probablistic concepts 
(where $y$ is a probabilistic but
scalar function of $x$), multiclass learning (where $y$ is a
multi- or vector-valued but deterministic function of $x$), and
settings in which the output space associated with a classification
may be large and complex. It is an instance of a more general framework
in which the distribution of multiple hidden variables ---
with unknown but parametric structural dependencies on observable inputs --- determines the distribution of
observable outputs. For the special case of a single binary hidden
variable, we provide the first formal learning guarantees in a
PAC framework.

As in the standard PAC model, we begin with an unknown binary function
or concept $c$ chosen from a known class $\cC$,\footnote{We leave the
  consideration of multi- or real-valued functions $c(x)$ to future
  work.} whose inputs $x$ are distributed according to an unknown and
arbitrary distribution. Now, however, the value $c(x)$ determines
which of two unknown probability distributions $P_{c(x)}$ govern the
distribution of $y$, where $P_0$ and $P_1$ are chosen from a known
class of distributions $\cP$. Thus $y$ is distributed according to a
mixture model, but the mixture component is given by a hidden
classifier $c$.  The learner does not see explicit labels $c(x)$, but
only the resulting $(x,y)$ pairs.  The goal is to learn a
\emph{hypothesis model} that consists of a hypothesis $h$ that is a
$\{0, 1\}$-valued function, and two probability distributions
$\widehat P_0$ and $\widehat P_1$ from the class $\cP$.  Given any
input $x$, the model will predict the vector $y$ to be drawn from the
distribution $\widehat P_{h(x)}$ (and hence \emph{predict with
  distribution} $\widehat P_{h(x)}$). Our objective is to minimize the
{\em conditional\/} Kullback-Leibler (KL) divergence
$\Ex{x}{\kl(P_{c(x)} || \widehat P_{h(x)})}$, rather than simply the
KL divergence to the mixture. We thus refer to our model as {\em
  Predicting with Distributions ($\pwd$)\/}.

One of our primary motivations is {\em composition\/} and {\em
  reducibility\/} across different learning models --- in this case,
models for classification and models for distribution learning. Within
the standard PAC (classification) model, there is a rich theory of
reducibility between specific learning problems~\citep{PW90,KV94},
between classes of learning problems~\citep{Schapire1990,Kearns98}, as
well as composition theorems allowing the creation of more complex
learning algorithm from simpler ones~\citep{KLV94}. Less common are
results allowing one to assemble algorithms with provable performance
guarantees from constituents that are solving different {\em types\/}
of learning problems. A natural starting point for such an
investigation is with the standard PAC supervised learning model, and
its distributional analogue~\citep{KMRRSS94}, since these models are
each already populated with a number of algorithms with strong
theoretical guarantees.

Our main technical interest is thus in conditions permitting {\em
  computationally\/} efficient learning algorithms composed of extant
classification and distribution learning algorithms.  Informally, our
results imply that for every concept class $\cC$ known to be PAC
learnable with classification noise~\citep{AL87}, and almost
every class $\cP$ known to be PAC learnable in the distributional
sense of~\citet{KMRRSS94}, $\pwd$ problems given by $(\cC,\cP)$ are
learnable in our framework.

\subsection{Our Results and Techniques}
Our results take the form of reductions from our model to algorithms
for PAC learning the concept class $\cC$ and the distribution class
$\cP$ separately.\footnote{Throughout the paper, all PAC learning
  algorithms (for both concept class $\cC$ and distribution class
  $\cP$) in our reduction runs in polynomial time, since we are
  primarily concerned with computational efficiency (as opposed to
  sample complexity).} 
The primary conceptual step is in identifying the natural
technical conditions that connect these two different classes of
learning problems. The centerpiece in this ``bridge'' is the notion of
a {\em distinguishing event\/} for two probability distributions
$P_0, P_1 \in \cP$, which is an event whose probability is
``signficantly'' (inverse polynomially) different under $P_0$ and
$P_1$, provided these distributions are themselves sufficiently
different.

Our first result shows that a distinguishing event can be used, via a
particular randomized mapping, to turn the observed $y$ into a noisy
binary label for the unknown concept $c$. This will serve as a
building block for us to combine efficient PAC learners from
classification and distribution learning.

We then use distinguishing events to provide two different reductions
of our model to PAC classification and distribution learning
algorithms.  In the ``forward'' reduction, we assume the distribution
class $\cP$ admits a small set of candidate distinguishing events. We
show that such candidate events exist and can be efficiently
constructed for the class of spherical Gaussians and product
distributions over any discrete domain.  By searching and verifying
this set for such an event, we first PAC learn $c$ from noisy
examples, then use the resulting hypothesis to ``separate'' $P_0$ and
$P_1$ for a distributional PAC algorithm for the class $\cP$. This
gives:

\begin{theorem}[{\em Informal Statement, Forward Reduction}]
\label{informalforward}
  Suppose that the concept class $\cC$ is PAC learnable under
  classification noise, and the distribution class $\cP$ is PAC learnable and
  admits a polynomial-sized set of distinguishing events. Then the
  joint class $(\cC,\cP)$ is $\pwd$-learnable.
\end{theorem}

In the ``reverse'' reduction, we instead first separate the
distributions, then use their approximations to learn $c$. Here we
need a stronger distribution-learning assumption, but no
assumption on distinguishing events. More precisely, we assume that {\em
  mixtures\/} of two distributions from $\cP$ (which is exactly what
the unconditioned $y$ is) are PAC learnable.  Once we have identified
the (approximate) mixture components, we show they can be used to
explicitly {\em construct\/} a specialized distinguishing event, which in turn lets us
create a noisy label for $c$. This leads our result in the reverse
reduction:

\begin{theorem}[{\em Informal Statement, Reverse Reduction}]
\label{informalreverse}
  Suppose that the concept class $\cC$ is PAC learnable under
  classification noise, and any mixture of two distributions from
  $\cP$ is PAC learnable. Then the joint class $(\cC,\cP)$ is
  $\pwd$-learnable.
\end{theorem}

In both reductions, we make central use of Le Cam's method to show
that any PAC concept or distribution learning algorithm must have a
certain ``robustness'' to corrupted data. Thus in both the forward and
reverse directions, by controlling the accuracy of the model learned
in the first step, we ensure the second step of learning will succeed.

Since practically every $\cC$ known to be PAC learnable can also be
learned with classification noise (either directly or via the
statistical query framework~\citep{Kearns98}, with parity-based
constructions being the only known exceptions), and the distribution
classes $\cP$ known to be PAC learnable have small sets of
distinguishing events (such as product distributions), and/or have
mixture learning algorithms (such as Gaussians),
our results yield efficient $\pwd$ algorithms for
almost all combinations of PAC classification and distribution learning
algorithms known to date.

\subsection{Related Works}

At the highest level, our model falls under the framework
of~\cite{H92}, which gives a decision-theoretic treatment of PAC-style
learning~\citep{V84} for very general loss functions; our model can be
viewed as a special case in which the loss function is conditional
log-loss given the value of a classifier.  Whereas~\cite{H92} is
primarily concerned with sample complexity, our focus here is on
computational complexity and composition of learning models.

At a more technical level, our results nicely connect two well-studied
models under the PAC learning literature. First, our work is related
to the results in PAC learning under classification
noise~\citep{AL87,D97,Kearns98}, and makes use of a result
by~\citet{RDM06} that established the equivalence of learning under
(standard) classification noise (CN) and under class-conditional
classification noise (CCCN). Our work also relies on the PAC model for
distribution learning~\citep{KMRRSS94}, including a long line of works
on learning mixtures of distributions (see
e.g.~\citet{Dasgupta99,AK01,VW04,FOS08}).  Our new model of $\pwd$
learning, in particular, can be viewed as a~\emph{composition} of
these two models.

Our model is also technically related to the one of
co-training~\citep{BM98} in that the input $x$ and the output $y$ give
two different views on the data, and they are conditionally
independent given the unknown label $z = c(x)$, which is also a
crucial assumption for co-training (as well as various other latent
variable models for inference and learning). However, our model is
also fundamentally different from co-training in two ways. First, in
our model, there is not a natural target Boolean function that maps
$y$ to the label $z$. For example, any outcome $y$ can be generated
from both distributions $P_0$ and $P_1$. In other words, just using
$y$ is not sufficient for identifying the label $z$. Second, our
learning goal is to predict what distribution the outcome $y$ is drawn
from given the input $x$, as opposed to predicting the unknown label
$z$.


\section{Preliminaries}\label{sec:prelim}
\subsection{Model: $\pwd$-Learning}
Let $\cX$ denote the space of all possible~\emph{contexts}, and $\cY$
denote the space of all possible~\emph{outcomes}. We assume that all
contexts $x\in \cX$ are of some common length $n$, and all outcomes
$y\in \cY$ are of some common length $k$. Here the lengths are
typically measured by the dimension; the most common examples for
$\cX$ are the boolean hypercube $\{0, 1\}^n$ and subsets of $\RR^n$
($\{0, 1\}^k$ and $\RR^k$ for $\cY$).

Let $\cC$ be a class of $\{0, 1\}$-valued functions (also
called~\emph{concepts}) over the context space $\cX$, and $\cP$ be a
class of probability distributions over the outcome space $\cY$. We
assume an \emph{underlying distribution} $\cD$ over $\cX$,
a~\emph{target concept} $c\in \cC$, and \emph{target distributions}
$P_0$ and $P_1$ in $\cP$. Together, we will call the tuple
$(c, P_0, P_1)$ the~\emph{target model}.

Given any target model $(c, P_0, P_1)$ and underlying distribution
$\cD$, our learning algorithm is then given sample access to the
following generative example oracle $\gen(\cD, c, P_0, P_1)$ (or
simply $\gen$). On each call, the oracle does the following
(see~\Cref{fig:k2} for an illustration):
\begin{enumerate}
\item Draws a context $x$ randomly according to $\cD$;
\item Evaluates the concept $c$ on $x$, and draws an outcome $y$
  randomly from $P_{c(x)}$;
\item Returns the context-outcome pair $(x,y)$.
\end{enumerate}

A~\emph{hypothesis model} is a triple
$T = (h, \widehat P_0, \widehat P_1)$ that consists of a
\emph{hypothesis} $h\in \cC$ and two \emph{hypothesis distributions}
$\widehat P_0$ and $\widehat P_1\in \cP$. Given any context $x$, the
hypothesis model predicts the outcome $y$ to be drawn from the
distribution $\widehat P_{h(x)}$ (or simply \emph{predicts with
  distribution} $\widehat P_{h(x)}$). The goal of our learning
algorithm is to output a hypothesis model with high accuracy with
respect to the target model, and the error of any model $T$ is defined
as
\[
  err(T) = \Ex{x\sim \cD}{\kl( P_{c(x)} || \widehat P_{h(x)})}
\]
where $\kl$ denotes Kullback-Leibler divergence (KL divergence).




Our model of~\emph{Predicting with Distributions learning
  ($\pwd$-learning)} is thus defined as follows.

\begin{definition}[$\pwd$-Learnable]
  Let $\cC$ be a concept class over $\cX$, 
  and $\cP$ be a class of distributions over $\cY$. We say
  that the joint class $(\cC, \cP)$ is \emph{$\pwd$-learnable} 
  if there exists an algorithm $\cL$ such that for any target concept
  $c\in \cC$, any distribution $\cD$ over $\cX$, and target
  distributions $P_0, P_1 \in \cP$ over $\cY$, and for any $\eps > 0$
  and $0 < \delta \leq 1$, the following holds: if $\cL$ is given
  inputs $\eps, \delta$ as inputs and sample access from
  $\gen(\cD, c, P_0, P_1)$, then $\cL$ will halt in time bounded by
  $\poly(1/\eps, 1/\delta, n, k)$ and output a triple
  $T = (h, \widehat P_0, \widehat P_1) \in \cC \times \cP \times \cP$
  that with probability at least $1 - \delta$ satisfies
  $err(T) \leq \eps$.
\end{definition}

Observe that the unconditional distribution over $y$ is a mixture of
the target distributions $P_0$ and $P_1$. In our model, it is not
enough to learn the mixture distribution (which is a standard problem
in learning mixtures of distributions). Our learning objective is to
minimize the expected \emph{conditional KL divergence}, which is more
demanding and in general requires a good approximation to the target
concept $c$ over $\cX$.

Also note that we have stated the definition for the ``proper''
learning case in which the hypothesis models lie in the target classes
$\cC$ and $\cP$. However, all of our results hold for the more general
case in which they lie in potentially richer classes $\cC'$ and
$\cP'$.

\ifnips
\else
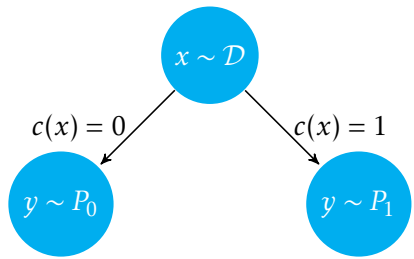
\begin{figure}[ht]
\begin{center}
\begin{tikzpicture}[->,>=stealth',shorten >=1pt,auto,node distance=2.8cm,
                    semithick, scale=0.8]
  \tikzstyle{every state}=[fill=cyan,draw=none,text=white]

  \node[state] (A)                    {$y\sim P_0$};
  \node[state]         (B) [above right of=A] {$x\sim \cD$};
  \node[state]         (C) [below right of=B] {$y\sim P_1$};

  \path (B) edge              node [left ]{$c(x) = 0$} (A)
        (B) edge              node [right]{$c(x) = 1$} (C);
\end{tikzpicture}
\end{center}
\caption{The generative model $\gen$: (1) first draw a context $x$
  from the underlying distribution $\cD$, (2) then evaluate the
  concept $c$ on $x$ and (3) draw the outcome $y$ from distribution
  $P_{c(x)}$. }
\label{fig:k2}
\end{figure}
\fi

\subsection{Related Learning Models}
We now discuss two learning models related to our setting (see the
appendix for formal definitions).


\paragraph{CN Learning}We first introduce PAC learning
under~\emph{classification noise (CN)}~\citep{AL87}. For any
~\emph{noise rate} $0 \leq \eta < 1/2$, consider the example oracle
$\EX_{\mathrm{CN}}^\eta(c, \cD)$ that on each call draws an example $(x, c(x))$
randomly according to $\cD$, then with probability $1 - \eta$ returns
the uncorrupted example $(x, c(x))$, and with probability $\eta$
returns the erroneous example $(x, \neg c(x))$. The concept class
$\cC$ is~\emph{CN learnable} if there exists a polynomial-time
algorithm that given sample access to $\EX_{\mathrm{CN}}^\eta$ finds a
hypothesis $h\in \cH$ that approximately minimizes the classification error:
$err(h) = \Pr_{x\sim \cD} [c(x) \neq h(x)]$.

\paragraph{CCCN Learning}{In a more general noise model
  called~\emph{Class-Conditional Classification Noise (CCCN)} proposed
  by~\cite{RDM06}, the example oracle $\EX_{\mathrm{CCCN}}^\eta$ has
  class-dependent noise rates --- that is, the noise rate $\eta_0$ for
  the negative examples ($c(x) = 0$) and the noise rate $\eta_1$ for
  the positive examples ($c(x) = 1$) may be different, and both below
  1/2.  Moreover, \cite{RDM06} show that any class that is learnable
  under CN is also learnable under CCCN. (See the appendix for a
  formal statement).}

\paragraph{Distribution Learning}{We also make use of results from for
  \emph{PAC learning probability distributions}~\citep{KMRRSS94}. A
  distribution class $\cP$ is efficiently learnable if there exists a
  polynomial-time algorithm that, given sample access to an unknown
  target distribution $P$, outputs an accurate distribution
  $\widehat P$ such that $\kl(P || \widehat P) \leq \eps$ for some
  target accuracy $\eps$. For any distribution $P\in \cP$ and any
  point $y\in \cY$, we assume that we can evaluate the probability
  (density) of $y$ assigned by $P$ (referred to as learning with an
  {\em evaluator\/} in~\cite{KMRRSS94}; see the appendix for the
  formal definition).  We will write $P(y)$ to denote the probability
  (or density) of point $y$, and write $P(E)$ to denote
  $\Pr_{y\sim P}[y\in E]$ for any measurable set $E\subset \cY$.}

To simplify our analysis, for the remainder of the paper we will make
the following assumption on the class $\cP$ to ensure that the
log-likelihood loss (or log-loss) is bounded in the domain
$\cY$. While this condition may not hold for some natural classes of
distributions (e.g.~Gaussians), it can be obtained using standard
procedures (for instance, by truncating, or mixing with a small amount
of the uniform distribution; see~\cite{FSO06} for an example).

\begin{assumption}[Boundedness Assumption]\label{ass:bounded}
  There exists a quantity $M$ that is upper bounded by $\poly(k)$ such
  that for any distribution $P\in \cP$ and any point $y\in \cY$, we
  have $\log(1/P(y)) \leq M$.
\end{assumption}

\section{CN Learning with Identified Distinguishing Events}\label{sec:cnlearn}

In this section, we will introduce a central concept to our
framework---\emph{distinguishing events}. Informally, an event
$E\subset \cY$ is distinguishing for distributions $P_0$ and $P_1$ if
it occurs with different probabilities under the measures of $P_0$ and
$P_1$. As a consequence, these events are informative about target
concept $c$ that determines which distribution the outcome $y$ is
drawn from. We will rely on such events to create a CCCN learning
instance for the target concept $c$. Thus, whenever the class $\cC$ is
learnable under CN (and hence learnable under CCCN by~\cite{RDM06}),
we can learn the target concept $c$ under the $\pwd$ model using a
distinguishing event.

\begin{definition}[Distinguishing Event]
  Let $P$ and $Q$ be distributions over the outcome space $\cY$, and
  let $\xi >0$. An event $E\subseteq \cY$ is $\xi$-distinguishing for
  distributions $P$ and $Q$ if $|P(E) - Q(E)| \geq \xi$.  We will call
  $\xi$ the~\emph{separation parameter} for such an event.
\end{definition}

We will now show that the knowledge of a distinguishing event between
$P_0$ and $P_1$ allows us to simulate an example oracle
$\EX_{\mathrm{CCCN}}^\eta$, and therefore we can learn the concept $c$
with a CCCN learner. The main technical problem here is to assign
noisy labels based on the distinguishing event so that noise rates
$\eta_0$ and $\eta_1$ of the oracle are strictly less than
$1/2$. 


Our solution is to construct a~\emph{randomized mapping} from the
event to the labels.\footnote{In the work of~\cite{BM98}, the authors
  showed that any CN learnable class is also learnable when the
  class-conditional noise rates satisfy $\eta_0 + \eta_1 <1$. Our
  construction here will imply a more general result---the class
  remains learnable when the noise rates satisfy
  $\eta_0 + \eta_1\neq 1$.}  Let us first introduce some
parameters. Let $E\subseteq \cY$ be a $\xi$-distinguishing event for
the distributions $P_0$ and $P_1$ for some $\xi \in (0, 1]$. We will
write $p = P_0(E)$ and $q = P_1(E)$.  Consider the following algorithm
$\lab(\hat p, \hat q, \xi)$ that takes parameters $\hat p$, $\hat q$
that are estimates for $p$ and $q$, and the separation parameter $\xi$
as inputs, and randomly creates noisy labels for $(x,y)$ pair drawn
from $\gen$: 
\begin{itemize}
\item Draw an example $(x, y)$ from the oracle $\gen$.
\item If $y\in E$, assign label $\ell = 1$ with probability $a_1$ and
  $\ell = 0$ with probability $a_0 = 1 - a_1$; Otherwise, assign label
  $\ell = 1$ with probability $b_1$ and $\ell = 0$ with probability
  $b_0 = 1 - b_1$, where
 \begin{align}
   a_0 = 1/2 + \frac{\xi(\hat p +  \hat q - 2)}{4(\hat q - \hat p)} \quad \mbox{ and } \quad
   b_0 = 1/2 + \frac{\xi(\hat p+\hat q)}{4(\hat q - \hat p)} 
\end{align}


\item Output the labeled example $(x, \ell)$.
\end{itemize}

It's easy to check that both vectors $(a_0, a_1)$ and $(b_0, b_1)$
form valid probabilities over $\{0, 1\}$ (see the appendix for a
proof).

As mentioned, we need to ensure the class-conditional noise rates to
be below $1/2$.  As a first step, we work out the noise rates of
$\lab$ in terms of the true probabilities $p$ and $q$, and show that
the ``estimated'' noise rates based on $\hat p$ and $\hat q$ are below
$(1/2 - \xi/4)$.

\begin{restatable}{lemma}{perfect}\label{lem:perfect}
  Given a fixed $\xi$-distinguishing event $E$, the class-conditional
  noise rates of $\lab$ are
\[
\eta_1 = \Pr[\ell = 0 \mid c(x) = 1] = q a_0 + (1 - q)b_0 \qquad \mbox{ and }\qquad
\eta_0 = \Pr[\ell = 1 \mid c(x) = 0] = p a_1 + (1 - p)b_1.
\]
Moreover, given any input estimates $(\hat p,\hat q)$ for $(p,q)$, the
parameters $a_0, a_1, b_0$ and $b_1$ satisfy:
\[
\hat q a_0 + (1 - \hat q)b_0 = \hat p a_1 + (1 - \hat p)b_1 \leq 1/2 - \xi/4.
\]
\end{restatable}

By Lemma~\ref{lem:perfect}, we know that as long as the input estimates
$\hat p$ and $\hat q$ are sufficiently close to $p$ and $q$, the noise
rates will be less than $1/2$.  To obtain such estimates, we will
guess the values of $p$ and $q$ on a grid of size
$\lceil 1/\Delta \rceil^2$ in the range of $[0, 1]^2$, where
$\Delta\in [0, 1]$ is some discretization parameter.  Note that
for some pair of values $i, j \in [\lceil 1/\Delta\rceil]$ and
$i\neq j$ such that the guesses
$(\hat p, \hat q) = (i\Delta, j\Delta)$ satisfies
\[
\hat p \in [p - \Delta, p + \Delta] \qquad \mbox{ and } \qquad
\hat q \in [q - \Delta, q + \Delta]
\]
Given such accurate guesses $\hat p$ and $\hat q$, we can then
guarantee low noise rates as derived below:

\begin{restatable}{lemma}{cccnoracle}\label{cccnoracle}
  Fix any $\Delta\in [0, 1]$. Suppose that the estimates $\hat p$ and
  $\hat q$ satisfy $|p - \hat p| \leq \Delta$ and
  $|q - \hat q|\leq \Delta$, then the class-conditional noise rates
  $\eta_0$ and $\eta_1$ for $\lab(\hat p, \hat q, \xi)$ are upper
  bounded by $1/2 - \xi/4 + \Delta$.
\end{restatable}

Thus, if we choose the discretization parameter $\Delta$ to be below
$\xi/4$, then the algorithm $\lab(\hat p, \hat q)$ is a valid example
oracle $\EX_{\mathrm{CCCN}}^\eta$ for some pair of guess
estimates. Furthermore, if we apply the corresponding CCCN learning
algorithm to the instantiations of $\lab(\hat p, \hat q)$ over all
guesses $(\hat p, \hat q)$, the output list of hypotheses is then
guaranteed to contain an accurate one.

\begin{restatable}{lemma}{event}\label{lem:1event}
  Let $\eps, \delta\in (0, 1)$. Suppose that the concept class $\cC$
  is CN learnable, and there exists an identified $\xi$-distinguishing
  event $E$ for the two target distributions $P_0$ and $P_1$. Then
  there exists an algorithm $\cL_1$ such that when given
  $\eps, \delta, \xi$ and $E$ as inputs, it will halt in time bounded
  by $\poly(1/\eps, 1/\delta, 1/\xi, n)$, and with probability at
  least $1 - \delta$, output a list of hypotheses that contains some
  $h$ such that $err(h) \leq \eps$.
\end{restatable}

In the next two sections, we will use the algorithm
in Lemma~\ref{lem:1event} as a subroutine for learning the target concept
$c$ in the $\pwd$ framework.

\section{Forward Reduction}\label{sec:forward}

Now we will give our forward algorithmic reduction: first use a CN
learner to approximate the target concept $c$ sufficiently well to
separate the distributions $P_0$ and $P_1$, then learn each
distribution using a distribution learner.\footnote{We use the term
  ``forward'' to indicate that the reduction decomposes the learning
  process into the steps suggested by the generative model depicted in
  Figure~\ref{fig:k2}.} We will rely on the result
in~\Cref{sec:cnlearn} to learn $c$ with a CCCN learner, but we do not
assume the learner has a priori identified a distinguishing
event. Instead, we will assume that the distribution class $\cP$
admits a parametric class of distinguishing events of polynomial size,
which allows us to distinguish any two distributions in $\cP$ with
large KL-divergence.

\begin{assumption}[Parametric Class of Distinguishing
  Events]\label{ass:eventclass}
  There exists a parametric class of events $\cE(\cdot)$ for the
  distribution class $\cP$ 
  such that for any $\gamma >0$ and for any two probability
  distributions $P$ and $Q$ in $\cP$ with $\kl(P || Q) \geq \gamma$,
  the class of events $\cE(\gamma)$ contains a $\xi$-distinguishing
  event $E$ for $P$ and $Q$, where $\xi \geq 1/\poly(k,
  1/\gamma)$. Furthermore, $\cE(\gamma)$ can be computed in time
  $\poly(k, 1/\gamma)$ and the cardinality
  $|\cE(\gamma)| \leq \poly(k, 1/\gamma)$.
\end{assumption}

To illustrate the intuition of how to construct such class of
distinguishing events, we will give a simple example here. In the
appendix, we will extend the construction to work for the class of
spherical Gaussian distributions and product distributions over
discrete domains.

\paragraph{Simple Example}{Consider the outcome space
  $\cY = \{0, 1\}^k$ and the class of full-support product
  distributions $\cP$ over $\cY$. Let $P, Q\in \cP$ be two
  distribution such that $\kl(P || Q)\geq \gamma$. Under the
  boundedness condition in~\Cref{ass:bounded}, it can be shown that
  there exists some coordinate $l$ such that
  $|P^l - Q^l| \geq 1/\poly(k, 1/\gamma)$, where
  $P^l = \Pr_{y\sim P}[y_l = 1]$ and $Q^l = \Pr_{y\sim Q}[y_l =
  1]$. Therefore, for each coordinate $l$, the event that the
  coordinate $y_j$ is 1 is a candidate distinguishing event, so the
  class of events is simply
  $\cE = \{\mathbf{1}[y_l = 1] \mid l \in [k]\}$.  }



Here is our main result in the forward reduction.

\begin{theorem}[(Formal version
  of~\Cref{informalforward})]\label{thm:forwardmain}
  Under the~\Cref{ass:eventclass} that $\cP$ admits a parametric class
  of events $\cE$, the joint class $(\cC, \cP)$ is $\pwd$-learnable as
  long as the concept class $\cC$ is CN learnable, and the
  distribution class $\cP$ is efficiently learnable.
\end{theorem}

We will present our reduction in three key steps.
\begin{enumerate}
\item First, as a simple extension to~\Cref{sec:cnlearn}, we can learn
  a hypothesis $h$ with sufficiently small error assuming the class of
  events $\cE$ contains a distinguishing event for the distributions
  $P_0$ and $P_1$.

\item Suppose we have learned an accurate hypothesis $h$ from the
  first step, we can then use $h$ to separate outcomes $y$ drawn from
  $P_0$ and $P_1$, and apply the distribution learner to learn
  accurate distributions $\widehat P_0$ and $\widehat P_1$. This
  creates an accurate hypothesis model
  $\widehat T = (h, \widehat P_0, \widehat P_1)$.


\item Finally, we need to handle the case where the distributions
  $P_0$ and $P_1$ are arbitrarily close, and there is no
  distinguishing event for us to learn the concept $c$. We will show
  in this case it is not necessary to learn the target concept, and we
  can directly learn the distributions without relying on an accurate
  hypothesis $h$.
\end{enumerate}

The main technical challenge lies in the second and third steps, where
we will apply the distribution learner (for single distributions in
$\cP$) on samples drawn from a mixture of $P_0$ and $P_1$. To tackle
this issue, we will prove a robustness result for any distribution
learner --- as long as the input distribution is sufficiently close to
the target distribution, the output distribution by the learner
remains accurate.~\footnote{Our result actually extends to any PAC
  learning algorithm, and we omit the simple details.}

\subsection{CN Learning with a Class of Events}
As a first step in our reduction, we will simply
extend Lemma~\ref{lem:1event}: for each event $E$ in the event class $\cE$,
run the CCCN learner using $E$ as a candidate distinguishing event. If
the two target distributions $P_0$ and $P_1$ have large KL divergence,
then one of the output hypotheses $h$ will be accurate with respect to
$c$:

\begin{restatable}{lemma}{cccn}\label{cor:cccn}
  Let $\eps, \delta \in (0, 1)$ and $\gamma >0$. Suppose that the
  class $\cC$ is CN learnable, the class $\cP$ admits a parametric
  class of events $\cE$ (as in~\Cref{ass:eventclass}). If the two
  distributions $P_0$ and $P_1$ satisfy
  $\max\{\kl(P_0||P_1), \kl(P_1 || P_0)\} \geq \gamma$, then there
  exists an algorithm $\cL_2$ that given sample access to $\gen$ and
  $\eps, \delta, \gamma$ as inputs, runs in time
  $\poly(1/\eps, 1/\delta, 1/\gamma, n)$, and with probability at
  least $1 - \delta$ outputs a list of hypotheses $H$ that contains a
  hypothesis $h$ with error $err(h) \leq \eps$.
\end{restatable}

\subsection{Robustness of Distribution Learner}
Before we proceed to the next two steps of the reduction, we will
briefly digress to give a useful robustness result showing that the
class $\cP$ remains efficiently learnable even if the input
distribution is slightly perturbed. Our result relies on the
well-known~\emph{Le Cam's method}, which is a powerful tool for giving
lower bounds in hypothesis testing. We state the following version for
our purpose.\footnote{In the usual statement of Le Cam's method, the
  right-hand side of the inequality is in fact
  $1 - \|Q_0^m - Q_1^m\|_{tv}$, where $\|\cdot\|_{tv}$ denotes total
  variation distance. We obtain the current bound by a simple
  application of Pinsker inequality.}

\begin{restatable}{lemma}{neypear}[Le Cam's method (see
  e.g.~\cite{G15,yu1997assouad})]\label{lem:np}
  Let $Q_0$ and $Q_1$ be two probability distributions over $\cY$, and
  let $\cA\colon \cY^m \rightarrow \{0, 1\}$ be a mapping from $m$
  observations in $\cY$ to either $0$ or $1$. Then
\[
  \Pr_{\cA, Y^m\sim Q_0^m}[ \cA(Y^m) \neq 0 ] + \Pr_{\cA, Y^m\sim Q_1^m}[
  \cA(Y^m) \neq 1 ] \geq 1 - \sqrt{m \kl(Q_0 || Q_1) /2}
\]
where $Y^m \sim Q_\theta^m$ denotes an i.i.d. sample of size $m$ drawn
from the distribution $Q_\theta$.
\end{restatable}


The lemma above shows that any statistical procedure that determines
whether the underlying distribution is $Q_0$ or $Q_1$ based on $m$
independent observations must have high error if the two distributions
are too close. In particular, if their KL divergence satisfies
$\kl(Q_0||Q_1) \leq 1/m$, then the procedure has at least constant
error probability under measure $Q_0$ or $Q_1$. Now let's construct
such a procedure $\cA$ using any distribution learner $\cL$ for the
class $\cP$. Suppose the learner is $\eps$-accurate with high
probability when given sample of size $m$, and the distribution $Q_0$
is in the class $\cP$. Consider the following procedure $\cA$:
\begin{itemize}
\item Run the learning algorithm $\cL$ on sample $S$ of size $m$. If
  the algorithm fails to output a hypothesis distribution, output
  $1$. Otherwise, let $\widehat Q$ be the output distribution by $\cL$.
\item If $\kl(Q_0 || \widehat Q) \leq \eps$, output $0$; otherwise output
  $1$.
\end{itemize}

Note that if the sample $S$ is drawn from the distribution $Q_0$, then
$\cA$ will correctly output $0$ with high probability based on the
accuracy guarantee of $\cL$. This means the procedure has to err when
$S$ is drawn from the slightly perturbed distribution $Q_1$, and so
the learner will with constant probability output an accurate
distribution $\widehat Q$ such that $\kl(Q_0 || \widehat Q)\leq
\eps$. More formally:

\begin{restatable}{lemma}{stability}\label{thm:stability}
  Let $\eps > 0$, $\delta \in (0, 1/2)$ and $m\in \NN$. Suppose there
  exists a distribution learner $\cL$ such that for any unknown target
  distribution $P\in \cP$, when $\cL$ inputs $m$ random draws from
  $P$, it with probability at least $1 - \delta$ outputs a
  distribution $\widehat P$ such that $\kl(P||\widehat P)\leq \eps$.
  Then for any $Q_0\in \cP$ and any distribution $Q_1$ over the same
  range $\cY$, if the learner $\cL$ inputs a sample of size $m$ drawn
  independently from $Q_1$, it will with probability at least
  $1 - \delta'$ output a distribution $\hat Q$ such that
  $\kl(Q_0 || \widehat Q) \leq \eps$, where
  $\delta' = \delta + \sqrt{m \kl(Q_0||Q_1)/2}$.
\end{restatable}
\ifnips\else
\begin{proof}
  Consider the procedure $\cA$ constructed above that uses the learner
  $\cL$ as a subroutine.  
By the guarantee of the algorithm, we know that
$\Pr_{\cL, Y^m \sim Q_0^m}[\kl(Q_0 || \widehat Q) \leq \eps] \geq 1 - \delta.$
This means $$\Pr_{\cA, Y^m \sim Q_0^m}[ \cA(Y^m) \neq Q_0 ] \leq \delta.$$
By Lemma~\ref{lem:np}, we have
\[
  \Pr_{\cA, Y^m \sim Q_1^m}[ \cA(Y^m) \neq Q_1] \geq 1 - \sqrt{\frac{m}{2} \kl(Q_0 || Q_1)} - \delta.
\]
This in turn implies that with probability at least
$(1 - \delta - \sqrt{\frac{m}{2} \kl(Q_0 || Q_1)})$ over the draws of
$Y^m\sim Q_1^m$ and the internal randomness of $\cL$, the output
distribution $\widehat Q$ satisfies $\kl(P || \widehat Q) \leq \eps$.
\end{proof}
\fi

Therefore, if the KL divergence between the target distribution and
the input distribution is smaller than inverse of the (polynomial)
sample size, the output distribution by the learner is accurate with
constant probability. By using a standard amplification technique, we
can guarantee the accuracy with high probability:

\begin{restatable}{lemma}{robustness}\label{lem:robustness}
  Suppose that the distribution class $\cP$ is PAC learnable. There
  exist an algorithm $\cL_2$ and a polynomial
  $\Mp(\cdot, \cdot, \cdot)$ such that that for any target unknown
  distribution $P$, when given any $\eps > 0$ and $0< \delta \leq 1/4$
  as inputs and sample access from a distribution $Q$ such that
  $\kl(P||Q) \leq 1/(2 \Mp(1/\eps, 1/\delta, k))$, runs in time
  $\poly(1/\eps, 1/\delta, k)$ and outputs a list of distributions
  $\cP'$ that with probability at least $1 - \delta$ contains some
  $\widehat P\in \cP'$ with $\kl(P||\widehat P)\leq \eps$.
\end{restatable}

As a consequence, even when input sample distribution is slightly
``polluted'', we can still learn the target distribution accurately
with a small blow-up in the computational and sample complexity.

\subsection{Learning the Distributions with an Accurate Hypothesis}
Now we will return to the second step of our reduction: use an
accurate hypothesis $h$ and distribution learner for $\cP$ to learn
the two distributions $P_0$ and $P_1$. For any observation $(x,y)$
drawn from the example oracle $\gen$, we can use the hypothesis $h$ to
determine whether the outcome $y$ is drawn from $P_0$ or $P_1$, which
allows us to create independent samples from both
distributions. However, because of the small error of $h$ with respect
to the target concept $c$, the input sample is in fact drawn from a
mixture between $P_0$ and $P_1$. To remedy this problem, we will
choose a sufficiently small error rate for hypothesis $h$ (but still
an inverse polynomial in the learning parameters), which guarantees
that the mixture is close enough to either one of single target
distributions. We can then apply the result in Lemma~\ref{lem:robustness}
to learn each distribution, which together gives us a hypothesis model
$(h, \widehat P_0, \widehat P_1)$.

\begin{restatable}{lemma}{separate}\label{lem:separate}
  Suppose that the distribution class $\cP$ is efficiently
  learnable. Let $\eps > 0, 0< \delta \leq 1$ and $h\in \cH$ be an
  hypothesis. Then there exists an algorithm $\cL_3$ and a polynomial
  $r(\cdot, \cdot, \cdot)$ such that when given $\eps$, $\delta$ and
  $h$ as inputs, $\cL_3$ runs in time bounded by
  $\poly(1/\eps, 1/\delta, k)$, and outputs a list of probability
  models $\cT$ such that with probability at least $1 - \delta$ there
  exists some $\widehat T\in \cT$ such that $err(\widehat T)\leq \eps$, as
  long as the hypothesis $h$ satisfies
  $err(h) \leq 1/r(1/\eps,1/\delta, k)$.
\end{restatable}

\subsection{Directly Applying the Distribution Learner}
In the last step of our forward reduction, we will consider the case
where the two target distributions $P_0$ and $P_1$ are too close to
admit a distinguishing event, and so we will not be able to learn the
target concept $c$ as in the first step. We show that in this case
learning $c$ is not necessary for obtaining an accurate probability
model --- we can simply run the robust distribution learner developed
in Lemma~\ref{lem:robustness} over the samples drawn from the mixture to
learn single distribution. 

We will first define the following notion of \emph{healthy mixture},
which captures the mixture distributions with non-trivial weights on
two sufficiently different components. This will also facilitate our
discussion in the reverse reduction.

\begin{definition}[Healthy Mixture]\label{def:healthy}
  Let $Q$ be mixture of two distributions $Q_0$ and $Q_1$ from the
  class $\cP$, and let $w_0$ and $w_1$ be the weights on the two
  components respectively. Then $Q$ is an $\eta$-\emph{healthy
    mixture} if both $\min\{w_0, w_1\}\geq \eta$ and
  $\max\{\kl(P_0||P_1), \kl(P_1 || P_0)\}\geq \eta$ hold. If one of
  the two conditions does not hold, we will call $Q$ an
  $\eta$-\emph{unhealthy mixture}.
\end{definition}

We now show that whenever the mixture distribution $P$ is unhealthy,
we can use the robust learner in Lemma~\ref{lem:robustness} to directly
learn a distribution $\widehat P$ for our prediction purpose (simply
always predict with $\widehat P$ regardless of the context $x$).  Note
that this not only includes the case where $P_0$ and $P_1$ are
arbitrarily close, but also the one where the weight on one component
is close to 0, which will be useful in~\Cref{sec:rev}.

\begin{restatable}{lemma}{direct}\label{lem:direct}
  Suppose that the distribution class $\cP$ is PAC learnable.  Let $P$
  be the unconditional mixture distribution over the outcomes $\cY$
  under the distribution $\gen$. Let $\eps > 0$ and
  $\delta \in (0, 1)$. Then there exists an algorithm $\cL_4$ and a
  polynomial $g(\cdot, \cdot, \cdot)$ such that when $\cL_4$ is given
  sample access to $\gen$ and $\eps, \delta$ as inputs, it runs in
  time bounded by $\poly(1/\eps, 1/\delta, k)$ and
  it will with probability at least $1 - \delta$, output a list of
  distributions $\cP'$ that contains $\widehat P$ with
  $\Ex{x\sim \cD}{\kl(P_{c(x)} || \widehat P)} \leq \eps$, as long as
  $P$ is an $\eta$-unhealthy mixture for some
  $\eta \leq 1/g(k,1/\eps, 1/\delta)$.
\end{restatable}

We will now combine the all the tools to provide a proof sketch
for~\Cref{thm:forwardmain} (see the appendix for details).

\begin{proof}[Proof Sketch for~\Cref{thm:forwardmain}]
  Our algorithm for $\pwd$ learning the joint class $(\cC, \cP)$ is
  roughly the following. First, we will make use
  of~\Cref{ass:eventclass} and obtain a set of candidate
  distinguishing events for the target distributions $P_0$ and $P_1$.
  We will run the CCCN learner to learn $c$ using each candidate event
  $E$ to generate noisy labels. This generates a list of hypotheses.
  We will use the hypotheses $h$ to separate the two distributions
  $P_0$ and $P_1$ and apply the algorithm in Lemma~\ref{lem:separate} to
  learn each distribution individually. This will give polynomially
  many hypothesis models $\hat T = (h, \hat P_0, \hat P_1)$.
  By Lemma~\ref{cor:cccn} and Lemma~\ref{lem:separate}, we know at least one of the models
  is accurate when $P_0$ and $P_1$ are sufficiently different.

  To cover the case where the two distributions are too close, we will
  use the algorithm in Lemma~\ref{lem:direct} to learn a list of
  distributions over $\cY$. In particular, the model
  $(h', \widehat P, \widehat P)$ is accurate for at least one of the
  output distribution $\widehat P$.
  
  Together, the two procedures above will give a list of polynomially
  many hypothesis models, at least one of which is guaranteed to be
  accurate. We will use the standard \emph{maximum likehood method} to
  output the model that minimizes empirical log-loss, and with high
  probability, this will be an accurate model.\footnote{See the
    appendix for the details and analysis of the maximum likelihood
    method in the $\pwd$ model.}
\end{proof}

We previously gave examples (such as product distributions and
special cases of multivariate Gaussians) that admit small classes of
distinguishing events, and to which Theorem~\ref{thm:forwardmain} can
be applied. There are other important cases --- such as general
multivariate Gaussians --- for which we do not know such
classes.\footnote{We conjecture that Gaussians do indeed have a small
  set of distinguishing events, but have not been able to prove it.}
However, we now describe a different, ``reverse'' reduction that
instead assumes learnability of mixtures, and thus is applicable to
more general Gaussians via known mixture learning
algorithms~\citep{Dasgupta99,AK01,FSO06}.

\section{Reverse Reduction}\label{sec:rev}
In our reverse reduction, our strategy is to first learn the two
distributions $P_0$ and $P_1$ sufficiently well, and then construct a
specialized distinguishing event to learn the target concept $c$ with
a CCCN learner.\footnote{We use the term ``reverse'' to indicate that
  the reduction decomposes the learning process into the steps
  suggested by the inverted generative model depicted in
  Figure~\ref{fig:k3}.}  We will make a stronger learnability
assumption on the distribution class $\cP$ --- we assume
a~\emph{parametrically correct} learner for any healthy mixture of two
distributions in $\cP$.


\begin{assumption}[Parametrically Correct Mixture Learning]\label{ass:mix}
  There exists a mixture learner $\lm$ and a polynomial $\rho$ such
  that for any $\eps >0, 0 < \delta\leq 1$, and for any $Z$ that is an
  $\eta$-healthy mixture of two distributions $Y_0$ and $Y_1$ from
  $\cP$, the following holds: if $\cL_M$ is given sample access to $Z$
  and $\eps, \delta > 0$ as inputs, $\cL_M$ runs in time
  $\poly(k, 1/\eps, 1/\delta)$ and with probability at least
  $1 - \delta$, outputs a mixture $\hat Z$ of distributions $\hat Y_0$
  and $\hat Y_1$ such that
  $\max\{\kl(Y_0 || \hat Y_0), \kl(Y_1 || \hat
  Y_1)\}\leq \eps$.
\end{assumption}

We remark that the assumption of~\emph{parametric correctness} is a
mild condition, and is satisfied by almost all mixture learning
algorithms in the literature (see
e.g.~\cite{Dasgupta99,FSO06,FOS08,HsuK13}). Also note that we only require
this condition when the healthy mixture condition
in~\Cref{def:healthy} is met. If the two either the two distributions
$Y_0$ and $Y_1$ are arbitrarily close or the mixture is extremely
unbalanced, we are not supposed to learn both components correctly.


\begin{theorem}[Formal Version
  of~\Cref{informalreverse}]\label{thm:rev}
  Suppose the class $\cC$ is CN learnable, the distribution class
  $\cP$ is efficiently learnable and satisfies the parametrically
  correct mixture learning assumption (\Cref{ass:mix}). Then the
  joint class $(\cC, \cP)$ is $\pwd$-learnable.
\end{theorem}

With the tools we develop for the forward reduction, the proof for
reverse reduction is straightforward.  There are essentially two cases
we need to deal with. In the first case where the mixture distribution
over $\cY$ is healthy, we can use the parametrically correct mixture
learner to learn the two target distributions, we can then use the
accurate approximations $\widehat P_0$ and $\widehat P_1$ to find a
distinguishing event for $P_0$ and $P_1$, which allows us to learn the
concept $c$ with a CCCN learner. In the case where the mixture
distribution is unhealthy and we cannot learn the components
accurately, we can again appeal to the robustness result we show using
Le Cam's method --- we can directly apply the learner for single
distributions and learn $P_0$ or $P_1$.

\begin{figure}
\begin{center}
\begin{tikzpicture}[->,>=stealth',shorten >=1pt,auto,node distance=2.8cm,
                    semithick]
  \tikzstyle{every state}=[fill=cyan,draw=none,text=white]

  \node[state]         (A)   {$x\sim \cD_l$};
  \node[state]         (B) [above right of=A] {$l \sim (w_0, w_1)$};
  \node[state]         (C) [below right of=B] {$y\sim P_l$};

  \path (B) edge              node [above left]{Draw $x$} (A)
        (B) edge              node [above right]{Draw $y$} (C);
\end{tikzpicture}
\end{center}
\caption{An alternative view of the generative model $\gen$: first
  draw a Bernoulli label $l$ with bias $w_1 = \Pr_{\cD}[c(x) = 1]$,
  then draw a context $x$ from the conditional distribution $\cD_l$ on
  $c(x) = l$, and an outcome $y$ from the distribution $P_l$. In the
  forward reduction, we first learn the concept $c$ over $\cX$ (which
  determines the label $l$), so we can separate the data and learn
  each distribution using a (single) distribution learner. In the
  reverse reduction, we will first use the mixture learner to learn
  both $P_0$ and $P_1$, and then use such information to obtain
  estimates for the label $l$ for learning the concept $c$.}
\label{fig:k3}
\end{figure}
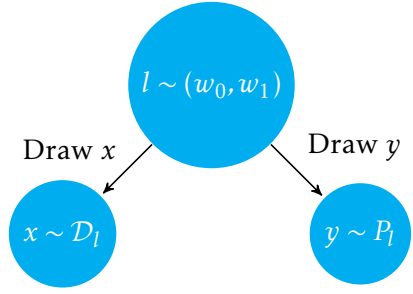

\subsection{CN Learning with a Mixture Learner}
Given any two distributions $P$, $Q$ over $\cY$ and a parameter
$\tau$, consider the event (or subset)
$$E({P, Q, \tau}) = \{y \in \cY \mid P(y) \geq 2^\tau \, Q(y) \}$$
We will first show that such subset is a distinguishing event for the
input distributions $P$ and $Q$ as long as the distributions $P$ and
$Q$ are sufficiently different.

\begin{restatable}{lemma}{admit}\label{lem:admit}
  Fix any $\gamma \in (0, 1]$. Suppose that $\kl(P || Q) \geq \gamma$,
  then $E(P, Q, \gamma/2)$ is a $(\gamma^2/(8M))$-distinguishing event
  for the distributions $P$ and $Q$.
\end{restatable}


Next, we show that even if we only have access to the approximate
distributions $\widehat P$ and $\widehat Q$, we can still identify a
distinguishing event for $P$ and $Q$, as long as the approximations
are accurate.

\begin{restatable}{lemma}{approxdist}\label{approxdist}
  Suppose that the distributions $P, \widehat P, Q, \widehat Q$ over
  $\cY$ satisfy that $\kl(P || \widehat P) \leq \alpha$,
  $\kl(Q || \widehat Q) \leq \alpha$, and $\kl(P || Q) \geq \gamma$
  for some $\alpha, \gamma \in (0, 1]$. Then the event
  $E(\widehat P, \widehat Q, (\gamma^2/(8M) - \sqrt{2\alpha})^2 )$ is
  a $\xi$-distinguishing event with
  $\xi \geq 1/\poly(1/\gamma, 1/\alpha, k) $ as long as
  $\gamma > 8M(\sqrt{2\alpha} + (8 M^2 \alpha)^{1/8})$.
\end{restatable}

Given these structural lemmas, we now know a way to construct a
distinguishing event based on approximations to the target
distributions $P_0$ and $P_1$.  We can then create a and use the
algorithm in Lemma~\ref{lem:1event} to learn the concept $c$, and in turn
compute a list of hypothesis models, one of which is guaranteed to be
accurate when the mixture distribution is healthy.

\begin{restatable}{lemma}{mixlearn}\label{lem:mixlearn}
  Suppose the class $\cP$ satisfies the parametric mixture learning
  assumption (\Cref{ass:mix}), the class $\cC$ is CN learnable, and
  mixture distribution over $\cY$ is $\gamma$-healthy for some
  $\gamma >0$.
Then there exists an algorithm $\cL$
that given $\eps, \delta$ and $\gamma$ as inputs and sample access from
$\gen$, halts in time bounded by
$\poly(1/\eps, 1/\delta, 1/\gamma, n, k)$, and with probability at
least $1 - \delta$, outputs a list of probability models $\cT$ that
contains some $\widehat T$ with $err(\widehat T) \leq \eps$.
\end{restatable}

Finally, to wrap up and prove~\Cref{thm:rev}, we also need to handle
the case where healthy mixture condition in~\Cref{def:healthy} does
not hold. We will again appeal to the robust distribution learner
in Lemma~\ref{lem:direct} to learn the distributions directly, and
construct hypothesis models based on the output distributions. To
guarantee that the output hypothesis model is accurate, we will again
use the maximum likelihood method to select the model with the minimum
empirical log-loss (formal proof deferred to the appendix).


\section{Future Work}

Despite the generality of our results and reductions, there remain some appealing
directions for further research. These include allowing the conditioning event
to be richer than a simple binary function $c(x)$, for instance multi- or even
real-valued. This might first entail the development of theories for noisy learning
in such models, which is well-understood primarily in the binary setting.

We also note that our study has suggested an interesting problem in pure
probability theory, namely whether general Gaussians permit a small class of distinguishing events.
\ifnips
\else
\fi



\paragraph{Acknowledgments}{We thank We thank Akshay Krishnamurthy and Shahin Jabbari for
  helpful discussions.}

\bibliographystyle{acmtrans}
\bibliography{./refs}

\appendix

\ifnips
\else
\newpage

\appendix

\section{Missing Details and Proofs}

\subsection{Missing Details in~\Cref{sec:prelim}}

\begin{definition}[CN Learnability~\citep{AL87}]
  Let $\cC$ be a concept class over $\cX$. We say that $\cC$ is
  efficiently learnable with noise (CN learnable) 
  if there exists a learning algorithm $\cL$ such that for any
  $c\in\cC$, any distribution $\cD$ over $\cX$, any noise rate
  $0\leq \eta < 1/2$, and for any $0 < \eps \leq 1$ and
  $0 < \delta \leq 1$, the following holds: if $\cL$ is given inputs
  $\eta_b$ (where $1/2 > \eta_b \geq \eta$), $\eps, \delta, n$, and is
  given access to $\EX_{\mathrm{CN}}^\eta(c, \cD)$, then $\cL$ will
  halt in time bounded by
  $\poly(1 / (1 - 2 \eta_b), 1/\eps, 1/\delta, n)$ and output a
  hypothesis $h\in \cH$ that with probability at least $1 - \delta$
  satisfies $err(h)\leq \eps$.
\end{definition}

\begin{lemma}[CN = CCCN~\citep{RDM06}]\label{lem:CCCN}
  Suppose that the concept class $\cC$ is CN learnable. 
  Then there exists an algorithm $\lc$ and a polynomial
  $\Mc(\cdot, \cdot, \cdot, \cdot)$ such that for every target concept
  $c\in \cC$, any $\eps, \delta\in (0, 1]$, for any noise rates
  $\eta_0, \eta_1\leq \eta_b <1/2$, if $L$ is given inputs
  $\eps, \delta, \eta_b$ and access to
  $\EX_{\mathrm{CCCN}}^\eta(c, \cD)$, then $L$ will halt in time
  bounded by $\Mc(1/(1 - 2\eta_b), 1/\eps, 1/\delta, n)$, and output
  with probability at least $1 - \delta$ a hypothesis $h$ with error
  $err(h) \leq \eps$. We will say that $\lc$ is an (efficient) CCCN
  learner for $\cC$ with sample complexity $\Mc$.
\end{lemma}

\begin{definition}[Evaluator~\citep{KMRRSS94}]
  Let $\cP$ be a class of distributions over the outcome space
  $\cY$. We say that $\cP$ has a efficient \emph{evaluator} if there
  exists a polynomial $p$ such that for any $n\geq 1$, and for any
  distribution $P\in \cP$, there exists an algorithm $E_P$ with
  runtime bounded by $\poly(k)$ that given an input $y\in \cY$ outputs
  the probability (density) assigned to $y$ by $P$. Thus, if
  $y\in \cY$, then $E_P(y)$ is the weight of $y$ under $P$. We call
  $E_P$ an~\emph{evaluator} for $P$.
\end{definition}

\subsection{Missing Proofs  in~\Cref{sec:cnlearn}}

\begin{claim}
  The values of $a_0$ and $b_0$ satisfy $a_0 , b_0\in [0, 1]$.
\end{claim}

\begin{proof}
  Without loss of generality, let's assume that $q \geq p + \xi$.
  Since $p+q \in [0, 2]$, we know that $a_0 \leq 1/2$ and we can write
\begin{align*}
a_0 = 1/2 +  \frac{\xi(p + q-2)}{4(q - p)} \geq 1/2 - \frac{\xi}{2(q - p)}  \geq 1/2 - 1/2 \geq 0
\end{align*}
Similarly, we know that $b_0 \geq 1/2$ and we can write
\begin{align*}
b_0 = 1/2 + \frac{\xi(p+q)}{4(q - p)}  \leq 1/2 + \frac{\xi/2}{\xi} = 1
\end{align*}
This proves our claim. 
\end{proof}

\perfect*
\begin{proof}
We can derive the probabilities as follows
\begin{align*}
  \Pr[\ell = 0 \mid c(x) = 1] &= \Pr[(\ell = 0) \wedge (y\in E) \mid c(x) = 1]
                                + \Pr[(\ell = 0) \wedge (y\notin E) \mid c(x) = 1]\\
&= \Pr_\gen[y\in E \mid c(x) = 1]\Pr_{\lab}[ \ell = 0 \mid (y\in E) \wedge (c(x) = 1) ]\\
&+ \Pr_\gen[y\notin E \mid c(x) = 1]\Pr_{\lab}[ \ell = 0 \mid (y\notin E) \wedge (c(x) = 1) ]\\
&= \Pr_\gen[y\in E \mid c(x) = 1] a_0 + \Pr_\gen[y\notin E \mid c(x) = 1] b_0\\
&= q\, a_0 + (1- q) b_0
\end{align*}
Similarly, we can also show that
$\Pr[\ell = 1 \mid c(x) = 0] = p a_1 + (1 - p)b_1$.
For the second part of the statement, we can show
\begin{align*} 
\hat q a_0 + (1 - \hat q)b_0 = \frac{\hat q}{2} + \frac{\xi (\hat p + \hat q - 2) \hat q}{4 (\hat q - \hat p)}
+\frac{(1-\hat q)}{2} + \frac{\xi(\hat p+\hat q) (1- \hat q)}{4(\hat q - \hat p)} = 1/2 - \xi/4\\
\hat p a_1 + (1 - \hat p)b_1 = \frac{\hat p}{2} - \frac{\xi(\hat p + \hat q - 2)\hat p}{4(\hat q - \hat p)}
+ \frac{(1 - \hat p)}{2} - \frac{\xi(\hat p + \hat q)(1-\hat p)}{4(\hat q - \hat p)} = 1/2 - \xi/4
\end{align*}
which recovers our claim.
\end{proof}

\cccnoracle*

\begin{proof}
Since $a_0, a_1, b_0, b_1\in [0, 1]$, and by our assumption on the
accuracy of $\hat p$ and $\hat q$, we have
\begin{align*}
\eta_1 - (\hat q a_0 + (1 - \hat q)b_0) = ( q a_0 + (1 - q)b_0) - (
\hat q a_0 + (1 - \hat q)b_0) = (q - \hat q) (a_0 - b_0) \leq \Delta\\
\eta_0 - (\hat q a_1 + (1 - \hat q)b_1) = ( q a_1 + (1 - q)b_1) - (
\hat q a_1 + (1 - \hat q)b_1) = (q - \hat q) (a_1 - b_1)\leq \Delta
\end{align*}

The result of Lemma~\ref{lem:perfect} tells us that 
\[
\hat q a_0 + (1 - \hat q)b_0 = \hat p a_1 + (1 - \hat p)b_1 \leq 1/2 - \xi/4
\]
Therefore, we must also have
$\eta_0, \eta_1 \leq 1/2 - \xi/4 + \Delta$.
\end{proof}

\event*

\begin{proof}
  Since the concept class $\cC$ is CN learnable, by the result
  of~\cite{RDM06} we know there exists an efficient algorithm $\cA$
  that when given access to some example oracle
  $\EX_{\mathrm{CCCN}}^\eta$ with $\eta_0, \eta_1 \leq 1/2 - \xi/8$,
  outputs a hypothesis $h$ with error bounded $\eps$ with probability
  at least $1 - \delta$, halts in time
  $\poly(1/\eps, 1/\delta, 1/\xi, n)$.
  
  Now let parameter $\Delta = \xi/8$, and consider the algorithm: for
  each pair of values $(\hat p, \hat q) = (i \Delta, j\Delta)$ such
  that $i, j\in [\lceil 1/\Delta \rceil]$ and $i\neq j$, use the
  $\lab(\hat p, \hat q, \xi)$ to generate labeled examples, and run
  the algorithm $\cA$ with sample access to $\lab$; if the algorithm
  halts in time $p$ and outputs an hypothesis $\hat h$, store the
  hypothesis in a the list $H$. In the end, output the hypothesis
  list.

  By Lemma~\ref{cccnoracle}, we know for some guessed values of $p'$
  and $q'$, the algorithm $\lab(p', q', \xi)$ is an CCCN oracle with
  noise rates $\eta_0, \eta_1 \leq 1/2 - \xi/8$. Then by the guarantee
  of the learning algorithm, we know with probability at least
  $1 - \delta$, the algorithm will output an $\eps$-accurate
  hypothesis under these guesses.
\end{proof}

\subsection{Missing Proofs in~\Cref{sec:forward}}

\cccn*

\begin{proof}
  Consider the following algorithm. We will first use the oracle $\cE$
  with input parameter $\gamma$ to obtain a class of events
  $\cE(\gamma)$ that contains a $\xi$-distinguishing event $E^*$ with
  $\xi \geq \poly(\gamma, 1/n)$. Then for each event
  $E\in \cE(\gamma)$, we will run the algorithm $\cA$ in
  Lemma~\ref{lem:1event} with accuracy parameters $\eps$, $\delta$,
  separation parameter $\xi$, and $E$ as an hypothetical
  distinguishing event as input. For each event, the instantiation of
  algorithm $\cA$ will halt in polynomial time. Furthermore, when the
  input event is $E^*$ it will with probability at least $1 - \delta$
  outputs a list of hypotheses $H$ that contains a hypothesis $h$ such
  that $err(h) \leq \eps$ by the guarantee of Lemma~\ref{lem:1event}.
\end{proof}

\robustness*

\begin{proof}
  Let $\cL$ be a distribution learner that given a independent sample
  of size $m$ drawn from the unknown target distribution $P$, runs in
  time bounded by $\poly(1/\eps, 1/\delta, n)$ with probability at
  least $1 - \delta$, outputs a distribution $P'$ such that
  $\kl(P||P') \leq \eps$. By Lemma~\ref{thm:stability}, we know that with
  probability at least $(1/2 - \delta)\geq 1/4$, the algorithm can
  also output a distribution $P''$ such that $\kl(P|| P'') \leq \eps$ if the
  algorithm is given a sample of size $m$ drawn from the distribution
  $Q$.

  Let $r = \log_{3/4} (1/\delta)$. Now we will run the algorithm $r$
  times on $r$ independent samples, each of size $m$. Let $\cP'$ be
  the list of output hypothesis distributions in these runs. We know
  that with probability at least $1 - (1- 1/4)^{r} = 1 - \delta$,
  there exists a distribution $\widehat P \in \cP'$ such that
  $\kl(P || \widehat P) \leq \eps$.
\end{proof}

The following is a technical lemma that allows us to bound the KL
divergence between between a mixture distribution and one of its
component.

\begin{restatable}{lemma}{unhealthy}\label{lem:unhealthy}
  Let $P$ and $Q$ be two distributions over $\cY$ and $R$ be a mixture
  of $P$ and $Q$ with weights $w_p$ and $w_q$ respectively. Then we
  have $\kl(P||R) \leq w_q \kl(P||Q)$.
\end{restatable}

\begin{proof}
  Let $w_p$ and $w_q$ be the weights associated with $P$ and $Q$
  respectively in the mixture $R$. 
\begin{align*}
   \kl(P || R) &= \int_{y} P(y) \log\left(\frac{P(y)}{R(y)}\right) dy\\
             &=  \int_{y} (w_pP(y) + w_qP(y)) \log\left(\frac{w_pP(y) + w_q P(y)}{w_pP(y) + w_qQ(y)}\right) \,dy\\
  (\mbox{by the log-sum inequality})\quad &\leq \int_y\left(w_p P(y)\log\left( \frac{w_pP(y)}{w_pP(y)}\right) \right) \, dy+ \int_y\left(w_q P(y)\log\left( \frac{w_q P(y)}{w_q Q(y)}\right)\,  \right)\,dy\\
&=  w_q \kl(P || Q) 
\end{align*}
which proves our claim.
\end{proof}

\separate*

\begin{proof}
  \newcommand{\MP}{m_2} Our algorithm will first call the oracle
  $\gen$ for
  $N = C\, \MP(2/\eps, 4/\delta, k) \, \left(\frac{M^2}{\eps^2}
    \log(1/\delta)\right)$ times, where $C$ is some constant (to be
  determined in the following analysis) and $\MP$ is the polynomial
  upper bound for the runtime of the algorithm defined
  in Lemma~\ref{lem:robustness}.  Then the algorithm will separate these
  data points $(x,y)$'s into two samples, one for $h(x) = 0$ and the
  other for $h(x) = 1$. For each sample corresponding to $h(x) = j$,
  if the sample size is at least $m = \MP(2/\eps, 4/\delta)$, the run
  the learning algorithm $\cL_2$ in Lemma~\ref{lem:robustness} to the
  sample with target accuracy $\eps/2$ and failure probability
  $\delta/4$ and obtain a polynomial list of distributions $\cP_j$;
  otherwise, simply output a singleton list containing any arbitrary
  distribution in $\cP$.

  Let $j\in \{0, 1\}$ and $\pi_j = \Pr_{x\sim \cD}[h(x) = j]$.  Let us
  first consider the case where $\pi_j \geq \eps/(2M)$.  In order to
  invoke Lemma~\ref{lem:unhealthy}, we will upper bound the quantity
  $w_j \kl(P_j || \widehat P_j)$, where
  $w_j = \Pr_{x\sim\cD}[c(x) = j]$. We know that for some large enough
  constant $C$, we can guarantee with probability at least
  $1 - \delta/4$, we will collect at least $m$ observations with
  $h(x) = j$. Let $\eps_h = err(h)$, note that when we instantiate the
  learner $\cL_2$ on the sample with $h(x) =j$, the input distribution
  $I_j$ is a $(\eps_h, 1 - \eps_h)$-mixture of the distributions
  $P_{1-j}$ and $P_j$. Then there exists a polynomial $r$ such that if
  $err(h) \leq 1/r(1/\eps,1/\delta, k)$, we can have the following
  based on Lemma~\ref{lem:unhealthy}
\[
  \kl(P_j || I_j) \leq \eps_h \kl(P || Q) \leq 1/\Mp(2/\eps , 4/\delta, k)
\]
where $\Mp$ is the polynomial defined in Lemma~\ref{lem:robustness}. This
means, the learning algorithm $\cL_2$ will with probability at least
$1 - \delta/4$, returns some distribution $\widehat P_j$ in the output
list such that $\kl(P_j || \widehat P_j) \leq \eps/2$, which implies that
$w_j \kl(P_j || \widehat P_j) \leq \eps/2$.  

Suppose that $\pi_j < \eps /(2M)$, then we know that no matter what
the distribution $\widehat P_j$ is, we have
$w_j \kl(P_j || \widehat P_j) \leq \frac{\eps}{2M} \, M = \eps/2$
by~\Cref{ass:bounded}.

Finally, our algorithm will output a list of probability models
$\cT= \{(h, \widehat P_0, \widehat P_1)\mid \widehat P_0\in \cP_0, \widehat P_1 \in
\cP_1\}$, such that with probability at least $1 - \delta$, there
exists some model $\widehat T = (h, \widehat P_0, \widehat P_1)\in \cT$ such
that
\[
err(T) = w_0 \kl(P_0 || \widehat P_0) + w_1 \kl(P_1 || \widehat P_1) \leq \eps,
\]
which recovers our claim.
\end{proof}

\direct*

\begin{proof}
  We first consider the case where the weight on one component is small, and
  without loss of generality assume that $w_1 \leq
  \eps/(4M\,m)$. By Lemma~\ref{lem:unhealthy} and~\Cref{ass:bounded}, we
  know that
  $$\kl(P_0 || R) \leq w_1 \kl(P_0 || P_1) \leq \frac{\eps}{2M \, m} \, M
  \leq 1/(2m).$$ By instantiating the algorithm
  in Lemma~\ref{lem:robustness} with parameters $(\eps/2, \delta)$, we know
  with probability $1- \delta$, there exists a hypothesis distribution
  $\widehat P$ in the output list such that
  $\kl(P_0 || \widehat P) \leq \eps/2$. Again by our~\Cref{ass:bounded}, we
  know $\kl(P_1 || \widehat P)\leq M$, so it follows that
\[
  \Ex{x\sim\cD}{\kl(P_{c(x)} || \widehat P)} = w_0 \kl(P_0 || \widehat P) + w_1
  \kl(P_1 || \widehat P) \leq \frac{\eps}{2} + \frac{\eps \, \kl(P_1 || \widehat P)}{2M\,
    m} \leq \eps.
\]
Next suppose that we are in the second case where
$\kl(P_0||P_1), \kl(P_1||P_0)\leq 1/(2m)$. We know
from Lemma~\ref{lem:unhealthy} that 
\[
\kl(P_0 || R) \leq w_1 \kl(P_0 || P_1) \leq 1/(2m) \quad\mbox{and, }\quad
\kl(P_1 || R) \leq w_0 \kl(P_1 || P_0) \leq 1/(2m)
\]

We will also apply the algorithm in Lemma~\ref{lem:robustness} which
guarantees with probability at least $1 - \delta$ that there exists a
hypothesis distribution $\widehat P$ in the output list $\cP'$ such that
$\kl(P_0 || \widehat P), \kl(P_1 || \widehat P) \leq \eps/2$, which implies that
\[
  \Ex{x\sim\cD}{\kl(P_{c(x)} || \widehat P)} = w_0 \kl(P_0 || \widehat P) + w_1
  \kl(P_1 || \widehat P) \leq \eps.
\]
Therefore, there exists a distribution $\widehat P$ in the output list
that satisfies our claim as long as we choose the polynomial $g$ such
that $g(1/\eps, 1/\delta, k) \geq \max\{2Mm/\eps, 2m\}$ for all
$\eps, \delta$ and $m$.
\end{proof}

\paragraph{Proof of~\Cref{thm:forwardmain}}{We will now combine the
  all the tools to prove~\Cref{thm:forwardmain}. First, consider the
  class of events $\cE(\gamma)$ with
  $\gamma = 1/g(1/\eps, 1/\delta, k)$ (specified
  in Lemma~\ref{lem:direct}). Then we will apply the CN algorithm $\cL_2$
  in Lemma~\ref{cor:cccn} to obtain a list $H$ of polynomially many
  hypotheses. For each $h\in H$, run the algorithm $\cL_3$ with $h$ as
  a candidate hypothesis. This will generate a list of a list of
  probability models $\cT$. If
  $\max\{\kl(P_0 || P_1), \kl(P_1 || P_0)\} \geq \gamma$, then $\cT$ is
  guaranteed to contain an $\eps$-accurate model with high probability
  (based on Lemma~\ref{cor:cccn} and Lemma~\ref{lem:separate}). Next, apply the
  distribution learner in Lemma~\ref{lem:direct} over the mixture
  distribution over $\cY$. If the algorithm outputs a distribution
  $\widehat P$, create a model $T'=(h_0, \widehat P, \widehat P)$, where
  hypothesis $h_0$ labels every example as negative. If
  $\max\{\kl(P_0 || P_1), \kl(P_1 || P_0)\} < \gamma$, we know $T'$ is
  $\eps$-accurate with high probability (based on Lemma~\ref{lem:direct}).  Finally, apply
  the~\emph{maximum likelihood} method to the list of models
  $\cT \cup \{T'\}$: draw a sample of polynomial size from $\gen$,
  then for each model $T\in \cT\cup\{T'\}$, compute the empirical
  log-loss over the sample, and output the model with the minimum log
  loss.  By standard argument, we can show that the output model is
  accurate with high probability.}

\subsection{Missing Proofs in~\Cref{sec:rev}}

\admit*

\begin{proof}
  Note that for any $y\in E$ such that $P(E) > 0$, we have
  $\log{\frac{P(y)}{Q(y)}}\leq M$ by~\Cref{ass:bounded}, and for any
  $y\notin E$, we also have
  $\log\left(\frac{P(y)}{Q(y)}\right) < \gamma/2$.

\begin{align*}
  \kl(P || Q) &= \int_{y\in\cY} P(y) \log{\frac{P(y)}{Q(y)}} dy\\
                 &= \int_{y\in E} P(y) \log{\frac{P(y)}{Q(y)}} dy + \int_{y\notin E}P(y) \log{\frac{P(y)}{Q(y)}} dy\\
                 &< P(E) M + (1 - P(E)) \frac{\gamma}{2} \\
                 &= \frac{\gamma}{2} + (M - \gamma/2)P(E) < \frac{\gamma}{2} + M \, P(E)
\end{align*}
Since we know that $\kl(P || Q) \geq \gamma$, it follows that
$P(E) > \frac{\gamma}{2M}$. Furthermore, 
\begin{align*}
  P(E) - Q(E) &= P(E) \left(1 - \frac{Q(E)}{P(E)} \right) \\
                  &\geq P(E) \left(1 - \sup_{y\in E}\frac{Q(y)}{P(y)} \right) \\
                  &\geq P(E) \left(1 - 2^{-\gamma/2}\right) \geq \frac{\gamma \, P(E)}{4} 
\end{align*}
where the last step follows from the fact that $1 - 2^{-a} \geq a/2$
for any $a\in [0, 1]$. It follows that
\[
P(E) - Q(E) > \frac{\gamma \, P(E)}{4} > \frac{\gamma}{2M} \, \frac{\gamma}{4} = \frac{\gamma^2}{8M},
\]
which proves our statement.
\end{proof}

\approxdist*

\begin{proof}
  Since we have both $\kl(P||\widehat P), \kl(Q||\widehat Q) \leq \alpha$, by
  Pinsker's inequality, we can bound the total variation distances
\[
\| P - \widehat P\|_{tv} \leq \sqrt{\alpha/2} \quad\mbox{ and, }\quad
\| Q - \widehat Q\|_{tv} \leq \sqrt{\alpha/2}.
\]
By Lemma~\ref{lem:admit} and the definition of total variation distance, we
know that 
\[
\|P - Q\|_{tv} = \sup_{E\subset \cY} |P(E) - Q(E)| \geq {\gamma^2/ (8M)}
\]
By triangle inequality, the above implies 
\[
\| \widehat P - \widehat Q\|_{tv} \geq \frac{\gamma^2}{8M} - \sqrt{2\alpha} \equiv b
\]
By Pinsker's inequality, we know that
$\|\widehat P - \widehat Q\|_{tv} \leq \sqrt{\kl(\widehat P || \widehat Q)/2}$.  It
follows that $\kl(\widehat P || \widehat Q) \geq 2b^2$. Consider the event
$E = E(\widehat P, \widehat Q, b^2)$. We know by Lemma~\ref{lem:admit} that $E$ is
a $(b^4/(2M))$-distinguishing event for distributions $\widehat P$ and
$\widehat Q$. Since both $\kl(P || \widehat P), \kl(Q || \widehat Q) \leq \alpha$,
we have
\[
|P(E) - \widehat P(E)| \leq  \|P(E') - \widehat P(E')\|_{tv} \leq \sqrt{\alpha/2}  \quad\mbox{ and, }\quad
|Q(E) - \widehat Q(E)| \leq  \|Q(E') - \widehat P(E')\|_{tv} \leq  \sqrt{\alpha/2}.
\]
Since $E$ is a $(b^4/(2M))$-distinguishing event for the distributions
$\widehat P$ and $\widehat Q$, this means $|\widehat P(E) - \widehat Q(E)| \geq (b^4/(2M))$,
and by triangle inequality, we have
\begin{align*}
  | P(E) - Q(E)| &= |(P(E) - \widehat P(E)) + (\widehat P(E) - \widehat Q(E)) +
  (\widehat Q(E) - Q(E))|\\
& \geq |\widehat P(E) - \widehat Q(E)| - |P(E) - \widehat P(E)| - |\widehat Q(E) - Q(E)|\\
& \geq (b^4/(2M)) - \sqrt{2\alpha}
\end{align*}
Note that if we have
$\gamma > 8M(\sqrt{2\alpha} + (8 M^2 \alpha)^{1/8})$, then we can guarantee both $b >0$ and  
$(b^4/(2M)) - \sqrt{2\alpha} > 0$.
\end{proof}

\mixlearn*

\begin{proof}
  We will first invoke the algorithm $\cL_M$ in~\Cref{ass:mix} so that
  with probability at least $1 - \delta/2$, the output approximations
  for the two components satisfy $\kl(P_0 || \widehat P_0) \leq \alpha$
  and $\kl(P_1 || \widehat P_1)\leq \alpha$ for some $\alpha$ that
  satisfies $\gamma > 8M(\sqrt{2\alpha} + (8 M^2 \alpha)^{1/8})$.
  This process will halt in time
  $\poly(1/\alpha, 1/\delta, 1/\gamma, k)$.

  By Lemma~\ref{lem:admit}, we know that the either event
  $E(\widehat P_0, \widehat P_1, \gamma/2)$ is a $\xi$-distinguishing
  event for $P_0$ and $P_1$ for some $\xi \geq 1/\poly(1/\gamma,n,k)$.
  Then we can use the CN learning algorithm $\cL_1$
  in Lemma~\ref{lem:1event} with the distinguishing event $E$ to learn a
  list of hypotheses $H$ under polynomial time, and there exists some
  $h\in H$ that is $\eps_1$ accurate, with
  $\eps_1 = 1/r(1/\eps, 1/\delta, k)$ (specified
  in Lemma~\ref{lem:separate}). For each hypothesis $h'\in H$, run the
  algorithm $\cL_3$ with $h'$ as the candidate hypothesis and $\eps$
  as the target accuracy parameter. By Lemma~\ref{lem:separate}, this will
  halt in polynomial time, and outputs a list of probability models
  $\cT$ such that one of which has error $err(\widehat T) \leq \eps$.
\end{proof}

\paragraph{Proof of~\Cref{thm:rev}}{ The algorithm consists of three
  steps. First, we will run the algorithm in Lemma~\ref{lem:mixlearn} by
  setting $\gamma = 1/g(1/\eps, \delta, k)$ (specified
  in Lemma~\ref{lem:separate}) and other parameters in a way to guarantee
  that whenever
  $\max\{\kl(P_0||\widehat P_0), \kl(P_1 || \widehat P_1)\} \geq
  \gamma$ and $\min\{w_0, w_1\}\geq \gamma$ both hold, the output list
  of models $\cT$ contains some $T$ that has error at most
  $\eps$. Next, we will directly apply the distribution learner
  in Lemma~\ref{lem:direct} so that when the healthy mixture condition is
  not met, the algorithm outputs a distribution $\widehat P$ such that
  $\Ex{x\sim \cD}{\kl(P_{c(x)} || \widehat P)}$. Lastly, similar to
  the final step in the forward reduction, we run the~\emph{maximum
    likelihood} algorithm to output the model in
  $\cT\cup \{(h_0, \widehat P, \widehat P)\}$ with the smallest
  empirical log-loss.  }

\section{Maximum Likelihood Algorithm}
In this section, we will formally define the maximum likelihood
algorithm, which is a useful subroutine to select an accurate
probability model from a list of candidate models. First, to give some
intuition, we show that the objective of minimizing
$\Ex{x\sim \cD}{\kl(P_{c(x)} || \widehat P_{h(x)})}$ is equivalent to
minimizing the expected log-losses. For any distribution $\widehat P$ over
$\cY$ and a point $r\in\cY$, the~\emph{log likelihood} loss (or simply
log-loss) is defined as $\loss(y, \widehat P) = - \log {\widehat
  P(y)}$. The~\emph{entropy} of a distribution $P$ over range $\cY$,
denoted $H(P)$, is defined as
\[
H(P) = \int_{y\in \cY} P(y) \log{\frac{1}{P(y)}} dy
\]
For any two distributions $P$ and $\widehat P$ over $\cY$, we could write
KL-divergence as
\begin{equation}\label{eq:kl}
  \kl(P||\widehat P) = \int_{y\in \cY} P(y) \log{\frac{1}{\widehat P(y)}}dy -
  H(P) = \Ex{y\sim P}{-\log \widehat P(y)} - H(P)
\end{equation}
which will be useful for proving the next lemma. 
\begin{lemma}
  Given any hypothesis $h\colon \cX\rightarrow \{0, 1\}$, and
  hypothesis distributions $\widehat P_0$ and $\widehat P_1$, we have
  \[
    \Ex{x\sim \cD}{\kl( P_{c(x)} || \widehat P_{h(x)})} = \Ex{x\sim
      \cD}{H(P_{c(x)})} - \Ex{(x,y)\sim \gen}{\log( \widehat P_{h(x)}(y))}
  \]
\end{lemma}

\begin{proof}
We can write the following
\begin{align*}
  \Ex{x\sim \cD}{\,\kl( P_{c(x)} || P_{h(x)}) } &= \Pr_{\cD}[c(x) = 1 , h(x) = 1] \,\kl(P_1 || \widehat P_1) + \Pr_{\cD}[c(x) = 1 , h(x) = 0] \,\kl(P_1 || \widehat P_0)\\
                                                &+ \Pr_{\cD}[c(x) = 0 , h(x) = 1] \,\kl(P_0 || \widehat P_1) + \Pr_{\cD}[c(x) = 0 , h(x) = 0] \,\kl(P_0 || \widehat P_0) \\
  (\mbox{apply Equation~\eqref{eq:kl}})\quad & = \Ex{x\sim \cD}{H(P_{c(x)})} - \sum_{(i,j)\in \{0, 1\}^2} \Pr_{\cD}[c(x) = i, h(x) = j]\,\Ex{y\sim P_i}{\log(\widehat P_j(y))}\\
&= \Ex{x\sim \cD}{H(P_{c(x)})} - \Ex{(x,y)\sim \gen}{\log( \widehat P_{h(x)}(y))}
\end{align*}
which proves our claim.
\end{proof}

Therefore, we could write
$err(T) = \Ex{x\sim \cD}{H(P_{c(x)})} - \Ex{(x,y)\sim \gen}{\log(
  \widehat P_{h(x)}(y))}$ for any model $T =(h, \hat P_0, \hat
P_1)$. Observe that $\Ex{x\sim \cD}{H(P_{c(x)})}$ is independent of
the choices of $(h, \widehat P_0, \widehat P_1)$, so our goal can also
be formulated as minimizing the expected log-loss
$\Ex{(x,y)\sim \gen}{\log( \widehat P_{h(x)}(y))}$. To do that, we
will use the following~\emph{maximum likelihood} algorithm: given a
list of probability models $\cT$ as input, draw a set of $S$ of
samples $(x,y)$'s from $\gen$, and for each
$T = (h, \widehat P_0, \widehat P_1)\in \cT$, compute the log-loss on
the sample
\[
\loss(S, T) = \sum_{(x,y)\in S} \loss(y, P_{h(x)}),
\]
and lastly output the probability model $\widehat T\in \cT$ with the
smallest $\loss(S, T)$. 

Our goal is to show that if the list of models $\cT$ contains an
accurate model $T$, the maximum likelihood algorithm will then output
an accurate model with high probability.

\begin{theorem}
  Let $\eps > 0$. Let $\cT$ be a set of probability models such that
  at least one model $T^*\in \cT$ has error $err(T^*)\leq
  \eps$. Suppose that the class $\cP$ also satisfies bounded
  assumption (in~\Cref{ass:bounded}).

  If we run the maximum likelihood algorithm on the list $\cT$ using a
  set $S$ of independent samples drawn from $\gen$. Then, with
  probability at least $1 - \delta$, the algorithm outputs some model
  $\hat T\in \cT$ such that $err(\hat T) \leq 4 \eps$ with
\[
  \delta \leq (|\cT| + 1) \exp\left( \frac{-2m\eps^2}{M^2}\right).
\]
\end{theorem}

To prove this result, we rely on the Hoeffding concentration bound.
\begin{theorem}\label{hoeffding}
  Let $x_1, \ldots, x_n$ be independent bounded random variables such
  that each $x_i$ falls into the interval $[a, b]$ almost surely. Let
  $X = \sum_i x_i$. Then for any $t > 0$ we have 
\[
  \Pr[X - \Ex{}{X} \geq t] \leq \exp\left(
    \frac{-2t^2}{n(b-a)^2}\right) \quad \mbox{ and }\quad
  \Pr[X - \Ex{}{X} \leq -t] \leq \exp\left(
    \frac{-2t^2}{n(b-a)^2}\right)
\]
\end{theorem}

\begin{proof}
  Our proof essentially follows from the same analysis of~\cite{FOS08}
  (Theorem 17). We say that a probability model $T$ is~\emph{good} if
  $err(T) \leq 4\eps$, and~\emph{bad} otherwise. We know that $\cT$ is
  guaranteed to contain at least one good model. In the following, we
  will write $H(\gen)$ to denote $\Ex{x\sim \cD}{H(P_{c(x)})}$.

  The probability $\delta$ that the algorithm fails to output some
  good model is at most the probability the best model $T^*$ has
  $\loss(S, T) \geq m\,(H(\gen) + 2\eps)$ or some bad model $T'$ has
  $\loss(S, T') \leq m\,(H(\gen) + 3\eps)$. Applying union bound, we
  get
  \[
    \delta \leq |\cT| \, \Pr[\loss(S, T') \leq m \, (H(\gen) + 3\eps) \mid err(T)\geq 4\eps]
+ \Pr[\loss(S, T^*) \geq m\, (H(\gen) + 2\eps)]
  \]
For each bad model $T'$ with $err(T') > 4\eps$, we can write
\begin{align*}
\Pr[\loss(S, T') \leq m(H(\gen) + 3\eps)] &= 
\Pr[\loss(S, T') \leq m(H(\gen) + 4\eps) - \eps m]\\
(\mbox{because } err(T') \geq 0) \qquad&\leq \Pr[\loss(S, T') \leq m(H(\gen) + err(T')) - \eps m]\\
&= \Pr[\loss(S, T') \leq \Ex{S\sim \gen^m}{\loss(S, T') - \eps}]\\
&\leq \exp\left(\frac{-2m\eps^2}{M^2} \right)
\end{align*}
where the last step follows from~\Cref{hoeffding}. Similarly, for the
best model $T^*$ with $err(T^*) \leq \eps$, we have the following
derivation:
\begin{align*}
\Pr[\loss(S, T^*) \geq m\, (H(\gen) + 2\eps)] &=
\Pr[\loss(S, T^*) \geq m\, (H(\gen) + \eps) + m\eps]\\
&\leq \Pr[\loss(S, T^*) \geq m\, (H(\gen) + err(T^*) + m\eps) ]\\
&= \Pr[\loss(S, T^*) \geq \Ex{S\sim \gen^m}{\loss(S, T^*)} + m\eps]\\
&\leq \exp\left(\frac{-2m \eps^2}{M^2} \right)
\end{align*}
Combining these two probabilities recovers the stated bound.
\end{proof}

In other words, as long as we have an $\eps$-accurate model in the
list, we can guarantee with probability at least $1 - \delta$ that the
output model has error $O(\eps)$ using a sample of size no more than
$\poly(k/\eps)\cdot \log(1/\delta)$.

\section{Examples of Distinguishing Events}
In this section, we give two distribution classes that admit
distinguishing event class of polynomial size.

\subsection{Spherical Gaussian}
We consider the class of spherical Gaussian in $\RR^k$ with fixed
covariance and bounded means. In particular, let
\[
  \cP = \{\cN(\mu, \Sigma) \mid \mu \in [0, 1]^k\}
\]
where $\Sigma$ is some diagonal covariance matrix in $\RR^{k\times k}$
such that the variance in each coordinate satisfy
$0 < \sigma_j^2\leq \sigma^2$ for some constant $\sigma > 1$.

\begin{theorem}
  There exists a parametric class of events $\cE(\cdot)$ for the
  distribution class $\cP$ of $k$-dimensional Spherical Gaussian
  such that for any $\gamma > 0$ and for any two probability
  distributions $P$ and $Q$ in the class $\cP$ such that
  $\kl(P || Q) \geq \gamma$, the class of events $\cE(\gamma)$
  contains an event $E$ that is an $\xi$-distinguishing event, where
  $\max\{1/\xi, |\cE(\gamma)|\} \leq \poly(k, 1/\gamma)$.
\end{theorem}

\begin{proof}
  Recall that the KL divergence of two multivariate Gaussian
  distributions $P$ and $Q$ with means $\mu, \mu'$ and covariance
  matrices $\Sigma_p, \Sigma_q$ can be written as
\[
  \kl(P || Q) = \frac{1}{2} \left(\text{tr}(\Sigma_q^{-1}\Sigma_p) +
    (\mu' - \mu)^\intercal \Sigma_q (\mu' - \mu) - k + \log\left(
      \frac{\det \Sigma_q}{\det \Sigma_p}\right) \right).
\]
For any two distributions $P$ and $Q$ in our class $\cP$, we can
simplify the KL divergence as
\[
  \kl(P || Q) \leq \frac{\sigma^2}{2} \|\mu - \mu'\|_2^2.
\]
Then $\kl(P || Q) \geq \gamma$ implies that there exists some
coordinate $j\in [k]$ such that
$|\mu_j - \mu_j'| \geq \sqrt{2\gamma/(k\sigma^2)}$. Note that the
marginal distributions of $P_j$ and $Q_j$ over the $j$-the coordinate
are $\cN(\mu_j, \sigma_j^2)$ and $\cN(\mu_j', \sigma_j^2)$
respectively. Without loss of generality, assume that
$\mu'_j < \mu_j$. Then for any value $t \in [\mu_j' , \mu_j]$, we have
\begin{align}\label{junky}
  P_j[ y \geq t] - Q_j[ y \geq t] &\geq P_j[y \in [t , \mu_j]].
\end{align}
Let $\Delta = \sqrt{2\gamma/(k\sigma^2)}$, and consider the
discretized set
$L(\gamma) = \{0, \Delta, \ldots, \lfloor 1/\Delta \rfloor
\Delta\}$. Then we know there exists a value $t'\in L$ such that
$t'\in L(\gamma)$ such that $t'\in [\mu_j', \mu_j]$ and
$\mu_j - t' \geq \Delta$. By~\Cref{junky}, we can write
\[
  P_j[ y \geq t'] - Q_j[ y \geq t'] \geq \frac{1}{2}\text{erf}(\Delta/(\sqrt{2}\sigma_j)) \geq \frac{1}{2}\text{erf}(\Delta/(\sqrt{2}\sigma))
\]
where $\text{erf}$ denotes the Gauss error function with
$\text{erf}(x) = \frac{2}{\sqrt{\pi}} \int_0^x e^{-a^2}\, da$ for
every $x\in \RR$. The Taylor expansion of the function is
\[
  \text{erf}(x) = \frac{2}{\sqrt{\pi}}\sum_{i = 0}^\infty \frac{(-1)^i
    x^{2i + 1}}{n!(2i + 1)} = \frac{2}{\sqrt{\pi}} \left( x -
    \frac{x^3}{3} + \frac{x^5}{10} - \frac{x^7}{42} \ldots\right)
\]
Therefore, for any $x \in [0, 1)$, there exists a constant $C$ such
that $\text{erf}(x/(\sqrt{2}\sigma)) / 2 \geq C \, x$. It follows that
\[
  P_j[ y \geq t'] - Q_j[ y \geq t'] \geq C \Delta.
\]
This means that the event of $(y_j \geq t')$ is a
$(C \Delta)$-distinguishing event for the two distributions $P$ and
$Q$. Therefore, for any $\gamma > 0$, we can construct the following
class of distinguishing events 
\[
  \cE(\gamma) = \{\mathbf{1}[y_j \geq t'] \mid j\in [k], t'\in
  L(\gamma)\}.
\]
Note that both $1/(C\Delta)$ and $|\cE(\gamma)|$ is upper bounded by
$\poly(1/\gamma, k)$, which recovers our claim.
\end{proof}

\subsection{Product Distributions over Discrete Domains}
Consider the space of $b$-ary cube $\cY = \{0, \ldots, b - 1\}^k$, and
the class of full-support product distributions $\cP$ over $\cY$:
distributions whose $k$ coordinates are mutually independent
distributions over $\{0, \ldots, b-1\}$. In particular, we assume that
there exists some quantity $M\leq \poly(k,b)$ such that for each
$P\in \cP$ and each coordinate $j$ and $y_j\in \{0, 1, \ldots b-1\}$,
we have $\log(1/P_j(y_j)) \leq M$. Now let's show that this class of
distributions admits a small class of distinguishing events as well.

\begin{theorem}
  There exists a parametric class of events $\cE(\cdot)$ for the
  production distribution class over the $b$-ary cube such that for
  any $\gamma > 0$ and for any two probability distributions $P$ and
  $Q$ in the class $\cP$ such that $\kl(P || Q) \geq \gamma$, the
  class of events $\cE(\gamma)$ contains an event $E$ that is an
  $\xi$-distinguishing event, where
  $\max\{1/\xi, |\cE(\gamma)|\} \leq \poly(k, b, 1/\gamma)$.
\end{theorem}

\begin{proof}
  In the following, we will write $P = P_1\times \ldots \times P_k$
  and $Q = Q_1\times\ldots \times Q_k$. Note that
\[
  \kl(P || Q) = \sum_{{j'}\in [k]} \kl(P_{j'} || Q_{j'}).
\]
Therefore $\kl(P||Q) \geq \gamma$ implies that there exists some
coordinate $j$ such that $\kl(P_j || Q_j)\geq \gamma/k$. This means
\[
  \sum_{y'_j \in \{0, \ldots , b-1\}} P_j(y'_j)
  \log\left(\frac{P_j(y'_j)}{Q_j(y'_j)} \right)\geq \gamma/k.
\]
This means there exists some $t\in \{0, \ldots, b-1\}$ such that
$P_j(t) \log(P_j(t) / Q_j(t)) \geq \gamma/(kb)$.  Recall that
$\log\left(P_j(t) / Q_j(t)\right) \leq M$, then we must have
$P_j(t) \geq \gamma/(kbM)$. Furthermore, since $P_j(t) \leq 1$, we
must also have $\log(P_j(t)/Q_j(t)) \geq \gamma/(kb)$. It follows
that
\[
  P_j(t) - Q_j(t) \geq P_j(t) \left(1 -
    \frac{Q_j(t)}{P_j(t)}\right) \geq \frac{\gamma}{kbM} \left(1 -
    2^{-\gamma/(kb)}\right) \geq \frac{\gamma}{kbM} \frac{\gamma}{2kb}
  = \frac{\gamma^2}{2(kb)^2M}
\]
where the last inequality follows from the fact that
$1 - 2^{-z} \geq z/2$ for any $z\in [0,1]$. Therefore, for any
$\gamma > 0$, the following class of events 
\[
  \cE(\gamma) = \{\mathbf{1}[y_j = t] \mid t\in \{0, 1,\ldots, b-1\},
  j\in [k]\}
\]
would contain a $\xi$-distinguishing event, and
$\max\{1/\xi, |\cE(\gamma)|\} \leq \poly(k, b, 1/\gamma)$.
\end{proof}

\fi

\end{document}